\documentclass[onecolumn, 12pt]{IEEEtran}

\usepackage[table]{xcolor}
\usepackage{times}
\usepackage{epsfig}
\usepackage[cmex10]{amsmath}
\usepackage{amsthm}
\usepackage{amsfonts}
\usepackage{graphicx}
\usepackage{amssymb}
\usepackage{amstext}
\usepackage{latexsym}
\usepackage{color,colortbl}
\usepackage{ifthen}
\usepackage{multirow}
\usepackage{verbatim}
\usepackage{array,tabularx}
\usepackage{arydshln}
\usepackage[mathscr]{euscript}
\usepackage{accents}
 \usepackage{cite}
\usepackage{hhline}
\usepackage{caption}
\usepackage{subcaption}
\usepackage{enumerate}
\usepackage{xcolor}
\usepackage{mathtools}
\usepackage{url}
\usepackage{xparse}
\usepackage{makecell}
\usepackage{varwidth}
\usepackage{bm}
\usepackage{arydshln}
\usepackage[top=1.5cm, bottom=1.5cm, left=1.5cm, right=1.5cm]{geometry}

\usepackage{comment}
\usepackage{wrapfig}

\usepackage{caption}
\usepackage{float}
\usepackage{booktabs}

\usepackage{dirtytalk}

\usepackage{mathtools}

\usepackage{booktabs}

\usepackage{multirow}

\newcolumntype{C}[1]{>{\centering\let\newline\\\arraybackslash\hspace{0pt}}m{#1}}

\newcommand{\dist}{\operatorname{dist}} 
\newcommand{\maj}{\operatorname{Maj}} 

\usepackage[normalem]{ulem}

\usepackage{amsmath,pgfplots,amssymb}
\usepackage{graphicx}

\newtheorem{theorem}{Theorem}

\newtheorem{lemma}{Lemma}

\theoremstyle{definition}
\newtheorem{example}{Example}
\newtheorem{definition}{Definition}
\newtheorem{corollary}{Corollary}

\theoremstyle{definition}
\newtheorem{remark}{Remark}

\theoremstyle{definition}

\newcommand{\interior}[1]{%
  {\kern0pt#1}^{\mathrm{o}}%
}

\usetikzlibrary{matrix,decorations.pathreplacing}

\newcommand{\restrictedto}[1]{|_{#1}}

\usepackage[utf8]{inputenc} 
\usepackage[T1]{fontenc}
\usepackage{url}
\usepackage{ifthen}
\usepackage{cite}

\newcommand\blfootnote[1]{%
  \begingroup
  \renewcommand\thefootnote{}\footnote{#1}%
  \addtocounter{footnote}{-1}%
  \endgroup
}

\begin{document}

\title{Differential Privacy for Binary Functions via Randomized Graph Colorings} 

\author{Rafael G. L. D'Oliveira\IEEEauthorrefmark{1}, Muriel Médard\IEEEauthorrefmark{1} and Parastoo Sadeghi\IEEEauthorrefmark{2}   \\ \IEEEauthorrefmark{1}RLE, Massachusetts Institute of Technology, USA\\ \IEEEauthorrefmark{2}SEIT, University of New South Wales, Canberra, Australia \\ Emails: \{rafaeld, medard\}@mit.edu, p.sadeghi@unsw.edu.au }

\maketitle

\begin{abstract} We present a framework for designing differentially private (DP) mechanisms for binary functions via a graph representation of datasets. Datasets are nodes in the graph and any two neighboring datasets are connected by an edge. The true binary function we want to approximate assigns a value (or true color) to a dataset. Randomized DP mechanisms are then equivalent to randomized colorings of the graph. A key notion we use is that of the boundary of the graph. Any two neighboring datasets assigned a different true color belong to the boundary.

Under this framework, we show that fixing the mechanism behavior at the boundary induces a unique optimal mechanism. Moreover, if the mechanism is to have a homogeneous behavior at the boundary, we present a closed expression for the optimal mechanism, which is obtained by means of a \emph{pullback} operation on the optimal mechanism of a line graph. For balanced mechanisms, not favoring one binary value over another, the optimal $(\epsilon,\delta)$-DP mechanism takes a particularly simple form, depending only on the minimum distance to the boundary, on $\epsilon$, and on $\delta$.
\end{abstract}

\blfootnote{The work of P.~Sadeghi was supported by the Australian Research Council Future Fellowship, FT190100429.}

\section{Introduction}
Since its inception, differential privacy (DP) \cite{dwork, dwork2014algorithmic} has become an important privacy-preserving tool in sharing information from datasets that contain sensitive information about individuals. A notable application of differential privacy was in the 2020 US Census privatization \cite{uscensus}, impacting hundreds of millions of people. 

The definition of differential privacy hinges upon the principle of \emph{neighboring} datasets -- those that differ in a single entry corresponding to one individual or sensitive feature. Roughly speaking, an $(\epsilon, \delta)$-DP mechanism aims to give the same randomized answer to a query from any two neighboring datasets with probabilities that are within $e^{\epsilon}$ multiplicative factor  of each other (modulo a small  additive constant $\delta$). Such a definition of DP is information-theoretic in the sense that it aims to limit the amount of information leakage about an individual in a dataset to an adversary with \emph{unbounded} computational power \cite{computationalDP}.\footnote{In contrast, computational DP \cite{computationalDP} relaxes this requirement and limits the information leakage to an adversary with finite computational power.}

The relationship between information-theoretic DP and local DP (LDP) \cite{LDP} with other notions of information-theoretic privacy have been studied. These include conditional mutual information \cite{cuff:2016} and maximal leakage \cite{MaxL2020}, which are under worst-case source distribution, as well as mutual information \cite{identifiability} and $\epsilon$-log-lift (also known as $\epsilon$-information-privacy or information density) \cite{Watchdog2019, Sadeghi2020ITW,PvsInfer2012}, which assume a given source distribution.

It can intuitively be understood that explicit DP conditions on neighboring datasets create topological privacy-preserving conditions into the fabric of the family of datasets of interest. In this paper, we propose to represent such topological $(\epsilon, \delta)$-DP conditions on discrete randomized mechanisms using graphs, where the vertices represent datasets and edges connect neighboring datasets. In this framework, a DP mechanism is a randomized coloring of the graph, subject to $(\epsilon, \delta)$-DP conditions. Crucially, we also consider utility via a \emph{true} coloring of the graph, where colors represent true values of the query function performed on a dataset. Any two neighboring datasets assigned a different true color belong to the graph boundary. To the best of our knowledge, a graph-based study of the tension between privacy and utility in the DP framework and corresponding optimal design of DP mechanisms is new.

As a first step towards a graph-based understating of this problem, we focus on binary functions. Applications include majority queries about voting or survey results, protecting participation of individuals in surveys, or simply crude quantized queries on whether a parameter of interest in a dataset is below or above a certain threshold. For a survey of applications of DP mechanisms for binary-valued functions see \cite{holahan_LP_DP}.

To illustrate, consider a case where three voters privately voted YES or NO to a sensitive matter. Considering all $2^3$ voting outcomes by three unique voters, Fig. \ref{fig: ex graphs a} shows the true majority function where blue means the majority voted YES and red means the majority voted NO. However, ignoring unique voters, these eight datasets can be compactly represented by (or collapsed on to) a \emph{line graph} comprising of four nodes, as in Fig. \ref{fig: ex graphs c}, where node $d_1$ means all three voted YES and node $d_2$ means any two people voted YES while the third voted NO - inversely, for nodes $d_3$ and $d_4$. 

From a mechanism design perspective, the line graph model for the majority function is much simpler to deal with. For $2n+1$ individuals, it reduces the complexity from $2^{2n+1}$ unique datasets to $2(n+1)$ datasets. But one might ask: is there any loss of \emph{optimality} in doing so? More broadly: is there a systematic and optimal way for \emph{importing or exporting} DP mechanisms across different families of datasets?  

\subsection{Summary of Results}
We illustrate our main results referring to Fig. \ref{fig: boundary graph}. We are interested in designing an optimal mechanism for the family of datasets represented by the graph in Fig. \ref{fig: boundary graph}(a). Here, optimal means the DP mechanism  dominates other mechanisms in terms of probability of truthful response (which we reasonably assume maximizes some utility function). 

\begin{itemize}
    \item We prove in Theorem \ref{theo: unique optimal} that if we fix the probability of giving the truthful response for each dataset in $\{h, \ell, q, n, i\}$ at the boundary, then there exists at most one optimal DP mechanism that satisfies these boundary conditions.
    \item In a \emph{boundary homogeneous} DP mechanism, only two parameters, $m_B$ and $m_R$, specify the probability of truthful response at blue and red boundary datasets,  respectively. Under this setting, we show through Definitions \ref{def:morph:general}, \ref{def:boundary:homog}, \ref{def:nB:nR:line}, \ref{def:morphto:nB:nR:line} and Theorems \ref{theo: morphism}, \ref{theo: general to line} that one can apply a color- and boundary- preserving morphism to obtain the line graph in Fig. \ref{fig: boundary graph}(b) with only two nodes in its boundary, optimally solve the $(\epsilon, \delta)$-DP mechanism over it, and pull it back to apply to Fig. \ref{fig: boundary graph}(a), while preserving optimality.
    \item In Theorem \ref{theo: optimal line mech}, we give a closed expression for the optimal $(\epsilon, \delta)$-DP mechanism for the line graph. Thus, we also obtain a closed expression for the optimal boundary homogeneous DP mechanism, via Theorem \ref{theo: general to line}.
    \item A mechanism is balanced if it is boundary homogeneous and $m_B = m_R$. The optimal balanced $(\epsilon, \delta)$-DP mechanism takes a very simple form. For any dataset $d$, the probability $P_e$ of giving the incorrect response (opposite to its true color) only depends on the shortest path to the nearest dataset of opposite color, $\Delta$, and the privacy parameters, $\epsilon$ and $\delta$:
    \begin{align*}
    P_e(d) = \max \left \{ \frac{e^\epsilon - 1 - \delta (e^{\epsilon (\Delta+1)} +e^{\epsilon \Delta} -2)}{e^{\epsilon \Delta} (e^\epsilon +1)(e^\epsilon -1)} , 0 \right \}.
\end{align*}
\end{itemize}

\section{Setting}

\begin{figure*}[!t]
\begin{subfigure}[t]{0.31\textwidth}
    \centering
    \begin{tikzpicture}[scale=1]

\draw[color=black, thick] (1,1) -- (3,1) -- (3,3) -- (1,3) -- cycle;
\draw[color=black, thick] (2,2) -- (4,2) -- (4,4) -- (2,4) -- cycle;
\draw[color=black, thick] (1,1) -- (2,2);
\draw[color=black, thick] (3,1) -- (4,2);
\draw[color=black, thick] (1,3) -- (2,4);
\draw[color=black, thick] (3,3) -- (4,4);

\draw[blue,fill=blue] (1,1) circle (.5ex);
\draw[blue,fill=blue] (3,1) circle (.5ex);
\draw[blue,fill=blue] (2,2) circle (.5ex);
\draw[blue,fill=blue] (1,3) circle (.5ex);

\draw[red,fill=red] (4,4) circle (.5ex);
\draw[red,fill=red] (2,4) circle (.5ex);
\draw[red,fill=red] (4,2) circle (.5ex);
\draw[red,fill=red] (3,3) circle (.5ex);

\node at (1,0.7) {\footnotesize{111}};
\node at (3,0.7) {\footnotesize{211}};
\node at (0.6,3) {\footnotesize{112}};
\node at (4.2,1.75) {\footnotesize{221}};
\node at (2.2,1.75) {\footnotesize{121}};
\node at (3.3,2.75) {\footnotesize{221}};
\node at (1.6,4) {\footnotesize{122}};
\node at (4.4,4) {\footnotesize{222}};

\end{tikzpicture}
    \caption{Hide vote value}
    \label{fig: ex graphs a}
\end{subfigure}
\begin{subfigure}[t]{0.31\textwidth}
    \centering
    \begin{tikzpicture}[scale=0.8]

\draw[color=black, thick] (1,1) -- (1,3) -- (3,5) -- (3,3) -- cycle;
\draw[color=black, thick] (1,3) -- (3,1) -- (5,3) -- (3,5);
\draw[color=black, thick] (3,3) -- (5,1) -- (5,3) -- (3,5);

\draw[blue,fill=blue] (1,1) circle (.6ex);
\draw[blue,fill=blue] (3,1) circle (.6ex);
\draw[blue,fill=blue] (1,3) circle (.6ex);
\draw[blue,fill=blue] (3,5) circle (.6ex);

\draw[red,fill=red] (5,1) circle (.6ex);
\draw[red,fill=red] (5,3) circle (.6ex);
\draw[red,fill=red] (3,3) circle (.6ex);

\node at (1,0.7) {\footnotesize{$\{1\}$}};
\node at (3,0.7) {\footnotesize{$\{2\}$}};
\node at (5,0.7) {\footnotesize{$\{3\}$}};
\node at (0.4,3) {\footnotesize{$\{1,2\}$}};
\node at (3,2.4) {\footnotesize{$\{1,3\}$}};
\node at (5.6,3) {\footnotesize{$\{2,3\}$}};
\node at (3,5.35) {\footnotesize{$\{1,2,3\}$}};

\end{tikzpicture}
    \caption{Hide if voted or not}
    \label{fig: ex graphs b}
\end{subfigure}
\begin{subfigure}[t]{0.31\textwidth}
    \centering
    \begin{tikzpicture}[scale=1]

\draw[color=black, thick] (1,1) -- (2,1) -- (3,1) -- (4,1);

\draw[blue,fill=blue] (1,1) circle (.5ex);
\draw[blue,fill=blue] (2,1) circle (.5ex);

\draw[red,fill=red] (3,1) circle (.5ex);
\draw[red,fill=red] (4,1) circle (.5ex);

\node at (1,0.7) {\footnotesize{$d_1$}};
\node at (2,0.7) {\footnotesize{$d_2$}};
\node at (3,0.7) {\footnotesize{$d_3$}};
\node at (4,0.7) {\footnotesize{$d_4$}};

\draw[white,fill=white] (2.5,-0.8) circle (.5ex);

\end{tikzpicture}
    \caption{Line graph for (a) or (b)}
    \label{fig: ex graphs c}
\end{subfigure}
\caption{Different types of neighborhood relations. (a) is explained in the main text. (b) shows an example where nodes represent which one of three individuals \{1\}, \{2\}, or \{3\} voted (voluntarily), whereas colors represent majority outcome assuming voters \{1\} and \{2\} always vote blue and \{3\} always votes red (ties go in favor of red). (c) is what we call a $(2,2)$-line graph in Definition \ref{def:nB:nR:line}. In Theorem \ref{theo: morphism}, we show that DP mechanisms in (c) can be transformed into DP mechanisms for both (a) and (b).}
\label{fig: ex graphs}
\end{figure*}
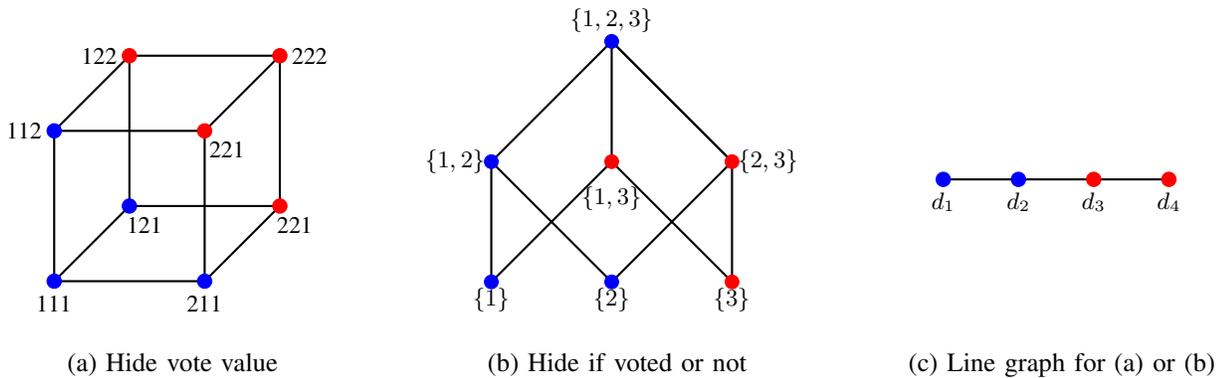 

We denote by $\mathcal{D}$ the family of datasets. We consider a symmetric neighborhood relationship in $\mathcal{D}$ where $d,d' \in \mathcal{D}$ are said to be neighbors if $d \sim d'$. We also consider a finite output space $\mathcal{V}$ which corresponds to the space over which the output of the queries lie. In this paper, we consider the case where $|\mathcal{V}|=2$ and that, without loss of generality, $\mathcal{V}=\{1,2\}$. 

A randomized mechanism, which we refer to as just a mechanism, is a random function $\mathcal{M}:\mathcal{D} \to \mathcal{V}$. We denote the set of all mechanisms of interest by $\mathfrak{M}$. In this paper, $\mathfrak{M}$ is the set of all $(\epsilon, \delta)$-DP mechanisms.

\begin{definition}\label{def: dp} Let $\epsilon, \delta \in \mathbb{R}$ be such that $\epsilon \geq 0$ and $0\leq \delta<1$. Then, a mechanism $\mathcal{M}:\mathcal{D} \to \mathcal{V}$ is $(\epsilon,\delta)$-differentially private if for any $d \sim d'$ and $S \subseteq \mathcal{V}$, we have  \[ \Pr [\mathcal{M}(d) \in \mathcal{S}] \leq e^\epsilon \Pr [\mathcal{M}(d') \in \mathcal{S}] + \delta .\] 
\end{definition}
For $|\mathcal{V}|=2$ $(\epsilon,\delta)$-differential privacy is equivalent to
\begin{align}\label{direct:equation:DP}
\Pr [\mathcal{M}(d) = v] \leq e^\epsilon \Pr [\mathcal{M}(d') = v] + \delta , \quad \forall v \in \mathcal{V}.
\end{align}

We consider a function $f:\mathcal{D} \to \mathcal{V}$ which we refer to as the true function. The goal is to approximate the true function $f$ by an $(\epsilon,\delta)$-differentially private mechanism $\mathcal{M}$. To measure the performance of the mechanism, i.e., how good the approximation is, a utility function $U:\mathfrak{M} \rightarrow \mathbb{R}$ must be defined, where $U[\mathcal M] \geq U[\mathcal M']$ means that the mechanism $\mathcal M$ performs better than $\mathcal M'$. In this work, we do not consider a specific utility function, but consider a general family of them. 

\begin{definition} \label{def: reasonable utility}
A utility function $U:\mathfrak{M} \rightarrow \mathbb{R}$ is \textit{reasonable} if $\Pr[\mathcal{M}(d)=f(d)] \geq \Pr[\mathcal{M'}(d)=f(d)]$ for every $d \in \mathcal{D}$ implies $U[\mathcal{M}] \geq U[\mathcal{M}']$. When this condition holds, we say that the mechanism $\mathcal{M}$ dominates $\mathcal{M}'$.
\end{definition}
\begin{remark}
This  notion of reasonable utility is relaxed enough not to impose unnecessary conditions on the application, but strong enough to capture some of the utility functions already proposed in the DP literature. The authors in \cite{ghosh_utility} considered a more restrictive notion of utility (negative of a loss function). A loss function $\ell: \mathcal{V} \times \mathcal V \to \mathbb{R}$ was called \emph{legal} in \cite{ghosh_utility} if for every true function value $i \in \mathcal{V}$ and mechanism response $j\in \mathcal{V}$, $\ell(i,j)$ depends only on $i$ and $|i-j|$ and is non-decreasing in $|i-j|$. This loss function can be used in numerical queries to measure the mean absolute error, where $\ell(i,j) = |i-j|$ or the mean square error, where $\ell(i,j) = |i-j|^2$. For categorical queries, by setting $\ell(i,j) = 0$ for $i = j$ and $\ell(i,j) = 1$ otherwise, one can measure the average binary loss function or Hamming distortion. Finding the optimal $(\epsilon, 0)$-LDP mechanism satisfying an upper bound on the expected Hamming distortion was studied in \cite{hamming_DP}. Indeed, for a given true function $f$, simultaneously maximizing the probability of truthful response across all datasets minimizes the expected Hamming distortion function: 
\[L[\mathcal{M}]\triangleq \sum_{d\in \mathcal{D}}\Pr(d)(1-\Pr(\mathcal{M}(d) = f(d)))\] regardless of the distribution on datasets $p(d)$. Therefore $U[\mathcal{M}] = 1-L(\mathcal{M})$ is a reasonable utility function.
\end{remark}

The notion of domination in Definition \ref{def: reasonable utility} induces a partial order on the set $\mathfrak{M}$ of all mechanisms. If a mechanism $\mathcal M$ dominates another $\mathcal M'$ then the first one outperforms the second for any reasonable utility function. It is not always the case that two mechanisms can be compared, even when restricted to a reasonable utility. We give an example below.

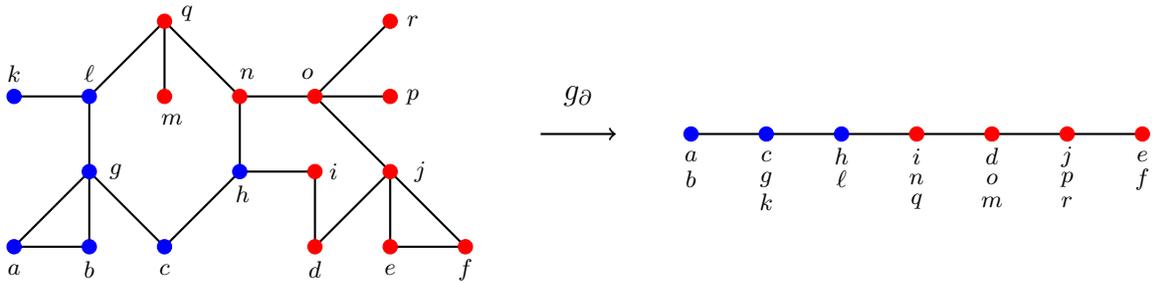
\begin{figure*}[!t]
    \centering
    \begin{tikzpicture}[scale=1]

\draw[color=black, thick] (1,1) -- (2,1) -- (2,2) -- cycle;
\draw[color=black, thick] (2,2) -- (3,1) -- (4,2) -- (4,3) -- (3,4) -- (2,3) -- cycle;
\draw[color=black, thick] (2,3) -- (1,3);
\draw[color=black, thick] (3,4) -- (3,3);
\draw[color=black, thick] (4,3) -- (5,3) -- (6,2) -- (7,1) -- (6,1) -- (6,2) -- (5,1) -- (5,2) -- (4,2) -- cycle;
\draw[color=black, thick] (5,3) -- (6,3);
\draw[color=black, thick] (5,3) -- (6,4);

\draw[blue,fill=blue] (1,1) circle (.5ex);
\node at (1,0.7) {{\footnotesize $a$}};
\draw[blue,fill=blue] (2,1) circle (.5ex);
\node at (2,0.7) {{\footnotesize $b$}};
\draw[blue,fill=blue] (3,1) circle (.5ex);
\node at (3,0.7) {{\footnotesize $c$}};
\draw[blue,fill=blue] (2,2) circle (.5ex);
\node at (2.35,2) {{\footnotesize $g$}};
\draw[blue,fill=blue] (4,2) circle (.5ex);
\node at (4.04,1.7) {{\footnotesize $h$}};
\draw[blue,fill=blue] (1,3) circle (.5ex);
\node at (1,3.3) {{\footnotesize $k$}};
\draw[blue,fill=blue] (2,3) circle (.5ex);
\node at (2,3.3) {{\footnotesize $\ell$}};

\draw[red,fill=red] (5,1) circle (.5ex);
\node at (5,0.7) {{\footnotesize $d$}};
\draw[red,fill=red] (6,1) circle (.5ex);
\node at (6,0.7) {{\footnotesize $e$}};
\draw[red,fill=red] (7,1) circle (.5ex);
\node at (7,0.7) {{\footnotesize $f$}};
\draw[red,fill=red] (5,2) circle (.5ex);
\node at (5.25,2) {{\footnotesize $i$}};
\draw[red,fill=red] (6,2) circle (.5ex);
\node at (6.4,2) {{\footnotesize $j$}};
\draw[red,fill=red] (3,3) circle (.5ex);
\node at (3.1,2.7) {{\footnotesize $m$}};
\draw[red,fill=red] (4,3) circle (.5ex);
\node at (4.1,3.3) {{\footnotesize $n$}};
\draw[red,fill=red] (5,3) circle (.5ex);
\node at (4.9,3.3) {{\footnotesize $o$}};
\draw[red,fill=red] (6,3) circle (.5ex);
\node at (6.3,3) {{\footnotesize $p$}};
\draw[red,fill=red] (3,4) circle (.5ex);
\node at (3.3,4.1) {{\footnotesize $q$}};
\draw[red,fill=red] (6,4) circle (.5ex);
\node at (6.3,4) {{\footnotesize $r$}};

\draw[->, color=black, thick] (8,2.5) -- (9,2.5);
\node at (8.5,3) {$g_{\partial}$};

\draw[color=black, thick] (10,2.5) -- (16,2.5);

\draw[blue,fill=blue] (10,2.5) circle (.5ex);
\node at (10,2.2) {{\footnotesize $a$}};
\node at (10,1.9) {{\footnotesize $b$}};
\draw[blue,fill=blue] (11,2.5) circle (.5ex);
\node at (11,2.2) {{\footnotesize $c$}};
\node at (11,1.9) {{\footnotesize $g$}};
\node at (11,1.6) {{\footnotesize $k$}};
\draw[blue,fill=blue] (12,2.5) circle (.5ex);
\node at (12,2.2) {{\footnotesize $h$}};
\node at (12,1.9) {{\footnotesize $\ell$}};

\draw[red,fill=red] (13,2.5) circle (.5ex);
\node at (13,2.2) {{\footnotesize $i$}};
\node at (13,1.9) {{\footnotesize $n$}};
\node at (13,1.6) {{\footnotesize $q$}};
\draw[red,fill=red] (14,2.5) circle (.5ex);
\node at (14,2.2) {{\footnotesize $d$}};
\node at (14,1.9) {{\footnotesize $o$}};
\node at (14,1.6) {{\footnotesize $m$}};
\draw[red,fill=red] (15,2.5) circle (.5ex);
\node at (15,2.2) {{\footnotesize $j$}};
\node at (15,1.9) {{\footnotesize $p$}};
\node at (15,1.6) {{\footnotesize $r$}};
\draw[red,fill=red] (16,2.5) circle (.5ex);
\node at (16,2.2) {{\footnotesize $e$}};
\node at (16,1.9) {{\footnotesize $f$}};

\end{tikzpicture}
    \caption{A graph morphism, which preserves neighboring relations, minimum distance (shortest path) to the boundary, and color.}
    \label{fig: boundary graph}
\end{figure*} 

\begin{example}
Consider the dataset $\mathcal{D}=\{1,2\}$ where $1 \sim 2$ and the true function $f:\mathcal{D} \rightarrow \mathcal{V}$ is such that $f(1) = 1$ and $f(2)=2$. Let $\mathcal{M}_1$ and $\mathcal{M}_2$ be the $(\log(2),0.1)$-DP mechanisms\footnote{In this paper, by $\log$ we mean the natural logarithm.} defined such that $\Pr[\mathcal{M}_1(1)=1] = 0.58$, $\Pr[\mathcal{M}_1(2)=2] = 0.76$, $\Pr[\mathcal{M}_2(1)=1] = 0.64$, and $\Pr[\mathcal{M}_2(2)=2] = 0.73$. Then, neither mechanism dominates the other. The reason for this is that there are reasonable utility functions which, for a mechanism $\mathcal{M} \in \mathfrak{M}$ might value the output of $\Pr[\mathcal{M}(1)=1]$ more than $\Pr[\mathcal{M}(2)=2]$, or vice-versa. Extreme cases of this are the reasonable utility functions $U[\mathcal{M}] = \Pr[\mathcal{M}(1)=1]$ and $U'[\mathcal{M}] = \Pr[\mathcal{M}(2)=2]$, both disagreeing on which of $\mathcal{M}_1$ or $\mathcal{M}_2$ is better.
\end{example}


\section{Differential Privacy as Randomized Graph Colorings}

In this section, we interpret differential privacy as a randomized graph coloring problem. The vertices of the graph are the datasets $d \in \mathcal{D}$. The edges of the graph are the neighboring relation on the datasets, i.e. two vertices $d,d' \in \mathcal{D}$ have an edge between them if $d \sim d'$. The graph is then a tuple $(\mathcal{D},\sim)$, which we often identify with the set $\mathcal{D}$ itself. 

The following transformation allows us to transport differentially private mechanisms from one setting to another.

\begin{definition}\label{def:morph:general}
A morphism from a family of datasets $\mathcal{D}_1$ to another family $\mathcal{D}_2$ is a function $g: \mathcal{D}_1 \rightarrow \mathcal{D}_2$ such that ${d\overset{1}{\sim}d'}$ implies in either $g(d)\overset{2}{\sim}g(d')$ or $g(d)=g(d')$, for every $d,d' \in \mathcal{D}_1$. 
\end{definition}

This notion is weaker than the classic graph homomorphism, which maps adjacent vertices to adjacent vertices, i.e. every graph homomorphism is a morphism, but not every morphism is a graph homomorphism. For example, the mapping from a graph with at least one edge to the graph with a single vertex is a morphism, but cannot be a graph homomorphism.

Morphisms allow us to transport mechanisms from the codomain to the domain via a pullback operation.

\begin{theorem} \label{theo: morphism}
Let $g: \mathcal{D}_1 \rightarrow \mathcal{D}_2$ be a morphism between two families of datasets and $\mathcal{M}_2 :\mathcal{D}_2 \rightarrow \mathcal{V}$ be an $(\epsilon,\delta)$-DP mechanism on $\mathcal{D}_2$. Then, the mechanism $\mathcal{M}_1:\mathcal{D}_1 \rightarrow \mathcal{V}$ given by the pullback operation $\mathcal{M}_1 = \mathcal{M}_2 \circ g$ is $(\epsilon,\delta)$-DP on $\mathcal{D}_1$.
\end{theorem}

\begin{proof}
Let $d,d' \in \mathcal{D}_1$ be such that $d \overset{1}{\sim} d'$. Then, 
\begin{align*}
    \Pr[\mathcal{M}_1 (d) = v] &= \Pr[\mathcal{M}_2 ( g(d) ) = v] \\
    &\leq e^\epsilon \Pr[\mathcal{M}_2 ( g(d') ) = v] + \delta \\
    &= e^\epsilon \Pr[\mathcal{M}_1 (d') = v] + \delta ,
\end{align*}
where the inequality follows from either $g(d)\overset{2}{\sim}g(d')$ or $g(d)=g(d')$.
\end{proof}

In Fig.~\ref{fig: boundary graph}, we show a morphism $g_\partial$ between a general graph and a line graph. In Theorem \ref{theo: general to line}, we use this same kind of morphism to obtain optimal $(\epsilon,\delta)$-DP mechanisms for a general class of graphs by pulling them back from optimal $(\epsilon,\delta)$-DP mechanisms on line graphs.

We now incorporate the true function we want to approximate into the graph. The true function $f:\mathcal{D} \to \mathcal{V}$ is equivalent to a coloring of the graph $\mathcal{D}$. We call the triple $(\mathcal{D},\sim,f)$ a colored graph, and often identify it with $\mathcal{D}$. We call a morphism $g: \mathcal{D}_1 \rightarrow \mathcal{D}_2$ such that $f_1 = f_2 \circ g$, a color preserving morphism. An $(\epsilon,\delta)$-DP mechanism is then a randomized coloring of the graph satisfying constraints related to the edges of the graph.

\begin{example} \label{ex: colored graphs}
Consider the dataset $\mathcal{D}=\{1,2\}^3$ where vertices are neighbors if they only differ in one entry, and the true function $\maj: \mathcal{D} \rightarrow \{1,2\}$ given by the majority function. If we assign colors to values such that $\maj(d)=1$ is blue and $\maj(d)=2$ is red, we obtain the graph in Fig.~\ref{fig: ex graphs a}. The function from Fig.~\ref{fig: ex graphs a} to Fig.~\ref{fig: ex graphs c} such that $111 \mapsto d_1$, $\{112,121,211\} \mapsto d_2$, $\{122,211,221\} \mapsto d_3$, and $222 \mapsto d_4$ is a color preserving morphism.
\end{example}

We define the following topological notions on our graphs.

\begin{definition}
The blue set is $B = \{ d \in \mathcal{D}:f(d)=1 \}$, corresponding to the color blue in our figures. The interior of $B$ is the set $\interior{B} = \{ d \in B : d \sim d' \Rightarrow d' \in B \}$ and its boundary is the set $\partial B = B - \interior{B}$. Replacing $B$ by $R$ above, we obtain the analogous red versions of the definitions. When referring to a single mechanism we denote the probabilities on the output by $B_d = \Pr[\mathcal{M}(d)=1]$ and $R_d = \Pr[\mathcal{M}(d)=2]$.

The distance between two points $d,d' \in \mathcal{D}$ is the number of edges in a shortest path connecting them, which we denote by $\dist (d,d')$. The distance from a point $d \in \mathcal{D}$ to a subset $A \subseteq \mathcal{D}$ is defined as $\dist (d,A) = \min_{d' \in A} \dist (d,d')$.
\end{definition}

Thus, if we consider the colored graph on the left of Fig.~\ref{fig: boundary graph}, the blue set is given by $B= \{a,b,c,g,h,k,\ell\}$, its interior by $\interior{B}=\{a,b,c,g,k\}$, and its boundary by $\partial B = \{h,\ell\}$.

\begin{remark} \label{rem: duality}
In this paper, we characterize mechanisms by how they behave on the blue set. Their behavior on the red set can then be derived by using analogous arguments. In general, our statements for the blue set imply in a dual version of them by replacing $B$ with $R$ and vice-versa. The dual of a colored graph is the graph with colors red and blue swapped.
\end{remark}

\section{Optimal Mechanisms}

In this section, we focus on finding optimal $(\epsilon,\delta)$-DP mechanisms for binary values. In Theorem \ref{theo: unique optimal}, we characterize the optimal mechanism in terms of its values at the boundary. Later, in Theorem \ref{theo: general graph}, we present a closed form for the optimal mechanism when the values at the boundary satisfy a homogeneity condition.\footnote{In what follows and in order to avoid cumbersome notation with $\max$ and $\min$ functions, every time a probability is less than zero we interpret it to be zero, and every time it is more than one we interpret it to be one.}

\begin{theorem} \label{theo: unique optimal}
Let $(\mathcal{D},\sim,f)$ be a colored graph and $m_d \in [0,1]$ be a fixed value for every $d \in \partial B$. Then, there exists at most one $(\epsilon,\delta)$-DP mechanism $\mathcal{M}:\mathcal{D} \rightarrow \mathcal{V}$ such that $B_d=m_d$, for every $d \in \partial B$.
\end{theorem}

\begin{proof}
We assume the subgraphs $B$ and $R$ are connected. If not, the following argument will hold for each connected component of $B$ and $R$. We also assume that there exists at least one mechanism which satisfies the $(\epsilon,\delta)$-DP constraints, otherwise our result trivially follows since their are no maximal $(\epsilon,\delta)$-DP mechanisms. 

Let $d,d' \in B$ be such that $d \sim d'$. Then, the $(\epsilon,\delta)$-DP conditions are given by
\begin{align} \label{eq: unique optimal dp1}
    B_d \leq e^\epsilon B_{d'} + \delta,
\end{align}
\begin{align}\label{eq: unique optimal dp2}
    B_{d'} \leq e^\epsilon B_{d} + \delta,
\end{align}
\begin{align}\label{eq: unique optimal dp3}
    1- B_d \leq e^\epsilon (1 -B_{d'}) + \delta,
\end{align}
\begin{align}\label{eq: unique optimal dp4}
    1 - B_{d'} \leq e^\epsilon (1-B_{d}) + \delta.
\end{align}
Assume, without loss of generality, that $B_d < B_{d'}$. Then, \eqref{eq: unique optimal dp1} and \eqref{eq: unique optimal dp4} are trivially satisfied. The remaining bounds, \eqref{eq: unique optimal dp2} and \eqref{eq: unique optimal dp3}, are both upper bounds on $B_{d'}$. Indeed, \eqref{eq: unique optimal dp3} is equivalent to $B_{d'}\leq \frac{e^\epsilon B_d+e^\epsilon +\delta -1}{e^\epsilon}$. Thus, $B_d$ and $B_{d'}$ are maximized together. Since $B$ is connected, this implies that all the $B_d$, for $d \in \interior{B}$, are maximized together. Define $\mathcal{M}$ to be the mechanism which maximizes all the $B_d$ simultaneously, for $d \in \interior{B}$, subject to the constraint that $\Pr[\mathcal{M}(d)=1]=m_d$, for every $d \in \partial B$.

We now consider the datasets $d \in R$. We note that, since the values at the border $d \in \partial B$ are already set, the maximization of the points in $\interior{B}$ does not affect the constraints on $R_d = \Pr[\mathcal{M}(d)=2]$. Thus, an argument analogous to the one above holds for the set $R$, i.e., all the $R_d$, for $d \in R$, can be maximized together. Thus, as above, we define $\mathcal{M}$ to be the mechanism which maximizes all the $R_d$ simultaneously, for $d \in R$. The mechanism $\mathcal{M}$ is then optimal.

\end{proof}

Thus, for every fixed values of $B_h$ and $B_\ell$ in the colored graph in the left of Fig.~\ref{fig: boundary graph}, there is either no $(\epsilon,\delta)$-DP mechanism or there is a unique maximal one. Moreover, the optimal mechanism can be found by simultaneously maximizing all the values in $B_x$ and $R_y$ for $x \in \interior{B}$ and $y \in R$. For example, if $B_h=0.6$ and $B_\ell = 0.7$, then the optimal $(\log(2),0.1)$-DP mechanism is such that $B_a = B_b = 1$, $B_g = B_k = 0.9$, $B_c=0.85$, $R_q = 0.7$, $R_i = R_n = 0.75$, $R_m=0.9$, $R_d=R_o = 0.925$, and $R_e = R_f = R_j = R_p = R_r = 1$. This can be checked by direct calculation of \eqref{direct:equation:DP} for all $d\sim d'$, showing that the $(\log(2),0.1)$-DP constraints are tightly satisfied. Another direct calculation shows that there is, however, no $(\log(1.1),0)$-DP mechanism for the same boundary conditions.

When the mechanism satisfies a homogeneity condition, we are able to find a closed expression, in Theorem \ref{theo: general graph}, for the optimal $(\epsilon,\delta)$-DP mechanism. This condition, we call boundary homogeneity, imposes the same probability of giving the truthful response at each same-color dataset of the boundary.

\begin{definition}\label{def:boundary:homog}
A mechanism $\mathcal{M}:\mathcal{D}\rightarrow \mathcal{V}$ is boundary homogeneous if, for every $d,d' \in \partial B$, it holds that $B_d = B_{d'}$. 
\end{definition}

Thus, a mechanism is boundary homogeneous if it acts the same across the boundary. For the voting example shown in Fig.~\ref{fig: ex graphs a}, a boundary homogeneous mechanism is agnostic to uniqueness of individuals, treating datasets $112$, $121$ and $211$ the same. In Theorem \ref{theo: general to line}, we show that the optimal boundary homogeneous $(\epsilon,\delta)$-DP mechanism of any colored graph can be obtained via a pullback of the optimal mechanism on a particular line graph.

\begin{definition}\label{def:nB:nR:line}
Let $n_B,n_R \in \mathbb{N}$. The $(n_B,n_R)$-line is the colored graph $(\mathcal{D},\sim,f)$ with datasets $\mathcal{D} = [1,n_B+n_R]$, neighboring relation $i \sim j$ if $|i-j|=1$, and true function such that $f([1,n_b])=1$ and $f([n_b+1,n_b+n_R])=2$.
\end{definition}

Examples include the $(2,2)$-line in Fig.~\ref{fig: ex graphs c}, the $(3,4)$-line in the right of Fig.~\ref{fig: boundary graph}, and the $(n_B,n_R)$-line in Fig.~\ref{fig: line tau}. We are particularly interested in the following $(n_B,n_R)$-line.

\begin{definition}\label{def:morphto:nB:nR:line}
Let $(\mathcal{D},\sim,f_1)$ be a colored graph and set $n_B = \max_{d \in B} (d,\partial B) +1$ and $n_R = \max_{d \in R} (d,\partial R)+1$. Then, the boundary graph of $\mathcal{D}$ is the $(n_B,n_R)$-line denoted by $(\mathcal{D}_\partial, \overset{\partial}{\sim},f_\partial)$. The boundary morphism is the color-preserving morphism $g_\partial : \mathcal{D} \rightarrow \mathcal{D}_\partial$ which maps $d \in B$ to $g_\partial(d) = n_B - \dist (d,\partial B)$ and $d \in R$ to $g_\partial(d) = n_B + 1 + \dist (d,\partial R)$.
\end{definition}

Fig.~\ref{fig: boundary graph} shows a colored graph on the left and its boundary graph on the right, with the explicit boundary morphism. Both colored graphs in Figs.~\ref{fig: ex graphs a} and \ref{fig: ex graphs b} have the $(2,2)$-line in Fig.~\ref{fig: ex graphs c} as their boundary graph. The morphism in Example \ref{ex: colored graphs} is a boundary morphism.

Our next result shows that the optimal boundary homogeneous $(\epsilon,\delta)$-DP mechanism of any colored graph can be obtained via a pullback of the optimal mechanism on its boundary graph.

\begin{theorem} \label{theo: general to line}
Let $(\mathcal{D}, \sim, f)$  be a colored graph and denote by $\mathcal{M}_\partial: \mathcal{D}_\partial \rightarrow \mathcal{V}$ the optimal $(\epsilon, \delta)$-DP mechanism on its boundary graph. Then, the pullback $\mathcal{M} = \mathcal{M}_\partial \circ g_\partial$ is the optimal boundary homogeneous $(\epsilon, \delta)$-DP mechanism on $\mathcal{D}$.
\end{theorem}

\begin{proof}
Let $n_B$ and $n_R$ be the parameters of the boundary graph, i.e. $\mathcal{D}_\partial$ is the $(n_B,n_R)$-line. By Theorem \ref{theo: unique optimal}, for each fixed $B_{n_B}$ (the probability of truthful response at the blue boundary dataset) there exists a unique maximal $(\epsilon,\delta)$-DP mechanism $\mathcal{M}_\partial$ on $\mathcal{D}_\partial$. By Theorem \ref{theo: morphism}, the morphism $g_\partial: \mathcal{D} \rightarrow \mathcal{D}_{\partial}$ induces an $(\epsilon,\delta)$-DP mechanism on $\mathcal{D}$ defined by $\mathcal{D}_{\partial} \circ g_\partial$. This mechanism is clearly boundary homogeneous. It follows from Theorem \ref{theo: unique optimal} that there is a unique optimal boundary homogeneous $(\epsilon,\delta)$-DP mechanism on $\mathcal{D}$. Let $\mathcal{M}$ be this mechanism. We show that $\mathcal{M} = \mathcal{D}_{\partial} \circ g_\partial$. 

Let $d \in \interior{B}$. Then, since $\mathcal{M}$ is optimal on $\mathcal{D}$, it holds that $\Pr[\mathcal{M}(d)=1] \geq \Pr[\mathcal{M}_\partial (g_\partial(d))=1]$. Let $d_0$ be the closest dataset to $d$ belonging to $\partial B$. Let $G = \{d, d_{\dist(d,\partial B)-1}, \ldots, d_0 \}$ be a set of datasets which form a shortest path from $d$ to $d_0$. Note that $g_{\partial} \restrictedto{G}$ is injective and thus has a left inverse, which we denote by $h: g_\partial (G) \rightarrow G$. Note that $h$ is a morphism and, therefore, by Theorem \ref{theo: morphism}, $\mathcal{M}\circ h$ is an $(\epsilon,\delta)$-DP mechanism on $\mathcal{D}_\partial$. It follows from Theorem \ref{theo: optimal line mech} that,  since $\mathcal{M}_\partial$ is the optimal mechanism on $\mathcal{D}_\partial$, then $\mathcal{M}_\partial \restrictedto{g_\partial (G)}$ is the optimal mechanism on $g_\partial (G)$. Thus, $\Pr[\mathcal{M}_\partial (g_\partial(d))=1] \geq \Pr[\mathcal{M}(d)=1]$, and, therefore, $\Pr[\mathcal{M}(d)=1] = \Pr[\mathcal{M}_\partial (g_\partial(d))=1]$.

An analogous argument holds for the red set $R$.
\end{proof}

Thus, finding the optimal boundary homogeneous $(\epsilon, \delta)$-DP mechanisms for general colored graphs is equivalent to finding them for the $(n_B,n_R)$-line. In Theorem \ref{theo: optimal line mech} we present a closed expression for the optimal $(\epsilon,\delta)$-DP mechanism on the $(n_B,n_R)$-line. We represent this mechanism in terms of the probability of the points in the blue set $B$ being red as a function of the distance to the boundary $\partial B$, denoted by $R_{n_B-i}$. We show that the mechanism is characterized by two possible behaviors, depending on a transition parameter, defined as follows.

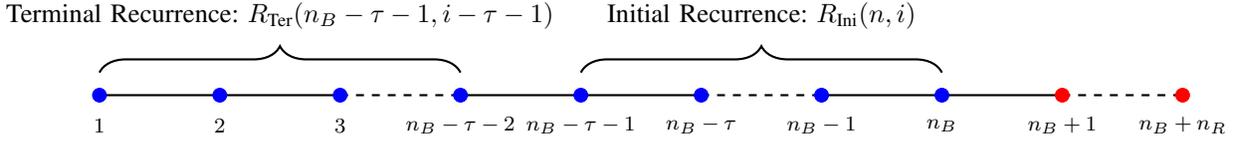
\begin{figure*}[!t]
    \centering
    \begin{tikzpicture}[scale=1.6]

\draw[color=black, thick] (1,0) -- (2,0) -- (3,0);
\draw[color=black, thick, dashed] (3,0) -- (4,0);
\draw[color=black, thick] (4,0) -- (5,0) -- (6,0);
\draw[color=black, thick, dashed] (6,0) -- (7,0);
\draw[color=black, thick] (7,0) -- (8,0) -- (9,0);
\draw[color=black, thick, dashed] (9,0) -- (10,0);

\draw[blue,fill=blue] (1,0) circle (.3ex);
\draw[blue,fill=blue] (2,0) circle (.3ex);
\draw[blue,fill=blue] (3,0) circle (.3ex);
\draw[blue,fill=blue] (4,0) circle (.3ex);
\draw[blue,fill=blue] (5,0) circle (.3ex);
\draw[blue,fill=blue] (6,0) circle (.3ex);
\draw[blue,fill=blue] (7,0) circle (.3ex);
\draw[blue,fill=blue] (8,0) circle (.3ex);

\draw[red,fill=red] (9,0) circle (.3ex);
\draw[red,fill=red] (10,0) circle (.3ex);

\node at (1,-0.25) {\scriptsize{$1$}};
\node at (2,-0.25) {\scriptsize{$2$}};
\node at (3,-0.25) {\scriptsize{$3$}};
\node at (4,-0.25) {\scriptsize{$n_B-\tau-2$}};
\node at (5,-0.25) {\scriptsize{$n_B-\tau-1$}};
\node at (6,-0.25) {\scriptsize{$n_B-\tau$}};
\node at (7,-0.25) {\scriptsize{$n_B-1$}};
\node at (8,-0.25) {\scriptsize{$n_B$}};
\node at (9,-0.25) {\scriptsize{$n_B+1$}};
\node at (10,-0.25) {\scriptsize{$n_B+n_R$}};

\draw [thick,decorate,decoration={brace,amplitude=10pt,mirror,raise=4pt},yshift=0pt] (8,0.1) -- (5,0.1) node [black,midway,yshift=0.9cm] {\small{Initial Recurrence: $R_{\text{Ini}}(n,i)$}};

\draw [thick,decorate,decoration={brace,amplitude=10pt,mirror,raise=4pt},yshift=0pt] (4,0.1) -- (1,0.1) node [black,midway,yshift=0.9cm] {\small{Terminal Recurrence: $R_{\text{Ter}} (n_B-\tau-1 , i-\tau-1)$}};

\end{tikzpicture}
    \caption{The $(n_B,n_R)$-line. In Theorem \ref{theo: optimal line mech}, we show that depending on the probability of being red $R_{n_B}$ at the blue boundary node $n_B$, there may be an initial recurrence phase for computing $R_{n_B-i}$ for its first $\tau+1$ adjacent nodes (Definitions \ref{def: tau} and \ref{def: initial recurrence}). After this possible initial phase, the terminal recurrence in Definition \ref{def: terminal recurrence} determines $R_{n_B-i}$.}
    \label{fig: line tau}
\end{figure*} 

\begin{definition} \label{def: tau}
Let $R_{n_B} \in [0,1]$ and $\epsilon, \delta \in \mathbb{R}_{\geq 0}$. Then, the transition parameter is defined as
\begin{align*}
    \tau = \left \lceil \frac{1}{\epsilon} \log \left( \frac{e^\epsilon+2\delta -1}{ (1-R_{n_B})(e^{3\epsilon}-e^{\epsilon}) + \delta (e^{2\epsilon}+e^{\epsilon}) } \right) \right \rceil,
\end{align*}
if $\epsilon>0$, and $\tau = -1$ if $\epsilon=0$.
\end{definition}

The initial behavior occurs when $i\leq \tau+1$. In this case, the probability of the mechanism outputting the color red is given by the following function.

\begin{definition} \label{def: initial recurrence}
The initial recurrence is given by
\begin{align*}
    R_{\text{Ini}}(n,i) &= 1-e^{ \epsilon i}(1-R_{n})-\frac{\delta (e^{i \epsilon}-1)}{e^{\epsilon}-1}.
\end{align*}
\end{definition}

The terminal behavior occurs when $i > \tau+1$ and is given by the following function.

\begin{definition} \label{def: terminal recurrence}
The terminal recurrence is given by
\begin{align*}
    R_{\text{Ter}}(n,i) =\frac{R_{n}}{e^{\epsilon i}} - \frac{\delta (e^{\epsilon i}-1)}{e^{\epsilon i}(e^\epsilon-1)}
\end{align*}
if $\epsilon>0$ and $R_{\text{Ter}} (n,i) =R_{n} - i \delta$, if $\epsilon=0$.
\end{definition}

These functions are obtained by solving the recurrences in the proof of Theorem \ref{theo: optimal line mech}, our next theorem. In this theorem, we present a closed form for the optimal $(\epsilon,\delta)$-DP mechanism on the $(n_B,n_R)$-line.

\begin{theorem} \label{theo: optimal line mech}
The unique optimal $(\epsilon , \delta)$-DP mechanism on the $(n_B,n_R)$-line with $\Pr[\mathcal{M}(n_B)=2]=R_{n_B}$ is such that
\begin{align*}
    R_{n_B-i} = \left\{\begin{matrix}
    R_{\text{Ini}} (n_B,i) & \text{if} \quad i \leq \tau +1,\\ 
    R_{\text{Ter}} (n_B-\tau-1 , i-\tau-1) & \text{if} \quad \tau+1 < i,
    \end{matrix}\right.
\end{align*}
for every $i \in [1,n_B-1]$.
\end{theorem}

\begin{proof}
Consider the $(\epsilon,\delta)$-DP conditions in \eqref{eq: unique optimal dp1}-\eqref{eq: unique optimal dp4} with the substitution $R_i = 1-B_i$. We are interested in minimizing the probability of giving the erroneous answer, $R_i$. Therefore, we consider two lower bounds on $R_i$ given by \eqref{eq: unique optimal dp1} and \eqref{eq: unique optimal dp4}, namely,
\begin{align} \label{eq: dp1}
    R_i \geq 1 - e^\epsilon + e^\epsilon R_{i+1} - \delta, 
\end{align}
and
\begin{align} \label{eq: dp2}
    R_i \geq \frac{R_{i+1}-\delta}{e^\epsilon}.
\end{align}
For each $i$, the largest of these bounds is the optimal choice for $R_i$. If $\epsilon=0$, then both bounds are the same and it is easy to check that the statement of the theorem holds. Thus, we assume $\epsilon>0$ in for the rest of this proof.

In Lemma \ref{lem: dp1>dp2}, we show that
\begin{align} \label{eq: dp1>dp2 ineq in proof}
    1 - e^\epsilon + e^\epsilon R_{i+1} - \delta > \frac{R_{i+1}-\delta}{e^\epsilon} ,
\end{align}
if and only if, 
\begin{align} \label{eq: ineq dp1>dp2}
    0 \leq \delta < e^{\epsilon} R_{i+1} + R_{i+1} - e^\epsilon .
\end{align}

Thus, every time \eqref{eq: ineq dp1>dp2} holds, the optimal $R_i$ is such that $R_i = 1 - e^\epsilon + e^\epsilon R_{i+1} - \delta $, i.e. making \eqref{eq: dp1} an equality. If we were to choose \eqref{eq: dp1} every time we would have the recurrence in Lemma \ref{lem: recur DP1}, with solution
\begin{align} \label{eq: initial recurrence}
    R_{n_B-i} = 1-e^{i \epsilon}(1-R_{n_B})-\frac{\delta (1-e^{i \epsilon})}{1-e^{\epsilon}} .
\end{align}

We find the first $i \in \mathbb{N}$ for which \eqref{eq: ineq dp1>dp2} does not occur. This happens when $e^\epsilon R_{n_B-i-1} + R_{n_B-i-1} - e^\epsilon \leq \delta$. Substituting $R_{n_B-i-1} = 1 - e^\epsilon + e^\epsilon R_{n_B-i} - \delta $ and rearranging, we obtain
\begin{align} \label{eq: ineq3}
    R_{n_B-i} \leq \frac{e^{2 \epsilon} + \delta e^\epsilon + 2 \delta + e^\epsilon -1}{e^{2 \epsilon} + e^\epsilon} .
\end{align}
Thus, whenever $R_{n_B-i}$ satisfies \eqref{eq: ineq3}, then $R_{n_B-i-1}$ will not satisfy \eqref{eq: ineq dp1>dp2}, so that the optimal choice for $R_{n_B-i-1}$ is equating it to $\eqref{eq: dp2}$. To find the first value such that this happens we substitute $R_{n_B-i}$ in \eqref{eq: ineq3} with its value in \eqref{eq: initial recurrence} to obtain
\begin{align*}
    1-e^{i \epsilon}(1-R_{n_B})-\frac{\delta (1-e^{i \epsilon})}{1-e^{\epsilon}} \leq \frac{e^{2 \epsilon} + \delta e^\epsilon + 2 \delta + e^\epsilon -1}{e^{2 \epsilon} + e^\epsilon}.
\end{align*}
Solving this for $i$ we obtain
\begin{align*}
    i  \geq \frac{1}{\epsilon} \log \left( \frac{1-e^\epsilon-2\delta}{ (1-R_{n_B})(e^{\epsilon}-e^{3\epsilon}) - \delta (e^{\epsilon}+e^{2\epsilon}) } \right).
\end{align*}
Thus, the smallest $i$ for which this occurs is $i = \tau$ as per Definition \ref{def: tau} (after multiplying both numerator and denominator by $-1$). 

To recap our argument, $R_{n_B - \tau}$ satisfies \eqref{eq: ineq3} which means that $R_{n_B - \tau -1}$ does not satisfy \eqref{eq: ineq dp1>dp2}. Thus, the initial recurrence applies up to $R_{n_B - \tau - 1}$. In other words, $R_{n_B-i} = R_{\text{Ini}} (n_B,i)$ for $i \leq \tau + 1$.

We now prove that for $i> \tau + 1$, the optimal choice is always \eqref{eq: dp2}. We do this by showing that if $R_{i+1}$ does not satisfy \eqref{eq: ineq dp1>dp2}, then $R_i$ does not either. Indeed, if $e^{\epsilon} R_{i+1} + R_{i+1} - e^\epsilon \leq \delta$, then, since $R_{i+1}$ does not satisfy \eqref{eq: ineq dp1>dp2}, $R_i = \frac{R_{i+1}-\delta}{e^\epsilon}$. Thus,
\begin{align*}
    e^\epsilon R_i + R_i - e^\epsilon &= R_{i+1} -\delta + \frac{R_{i+1}-\delta}{e^\epsilon} -e^\epsilon \\
    &= e^\epsilon R_{i+1} - e^\epsilon \delta + R_{i+1} - \delta - e^{2 \epsilon} \\
    &\leq - e^\epsilon \delta \leq \delta .
\end{align*}

Therefore, for $i> \tau+1$, the optimal $R_{n_B-i}$ is equating \eqref{eq: dp2}. This gives us a recurrence, which by Lemma \ref{lem: recur DP2}, has solution $R_{n_B-i} = R_{\text{Ter}} (n_B-\tau-1, i-\tau-1)$ for $i> \tau +1$.
\end{proof}

In the following example, we compute the optimal scheme for the $(4,3)$-line satisfying the boundary condition $R_4 = 0.8$.

\begin{example} \label{ex: (4,3)-line}
Consider the $(4,3)$-line with the boundary satisfying $R_4 = 0.8$, and with privacy parameters $\epsilon = \log (1.3)$ and $\delta = 0.1$. Then, $\tau = 1$, which means that $R_2$ and $R_3$ are calculated via the initial recurrence and $R_1$ via the terminal one. Performing this calculation we obtain, $R_3 = 0.64$, $R_2=0.432$, and $R_1 = 0.2553$.

We deal with the red set as noted in Remark \ref{rem: duality}. The dual of the $(4,3)$-line with the boundary satisfying $R_4 = 0.8$ is the $(3,4)$-line with boundary $R_4 = 0.2$. Then, $\tau = -2$, which means that $R_1$, $R_2$, and $R_3$ are calculated via the terminal recurrence. Performing this calculation we obtain, $R_3 = 1/13$, and $R_2 = R_1 = 0$. Thus, in the original $(4,3)$-line, the optimal mechanism satisfies $B_5 = 1/13$, and $B_6 = B_7 = 0$.
\end{example}

Combining Theorems \ref{theo: general to line} and \ref{theo: optimal line mech} we present a closed form for the optimal boundary homogeneous $(\epsilon , \delta)$-DP mechanism.

\begin{theorem} \label{theo: general graph}
Let $(\mathcal{D},\sim)$ be a set of datasets with a neighboring relation and $n_B = \max_{d \in B} \dist(d,\partial B) + 1$. Then, the optimal boundary homogeneous $(\epsilon , \delta)$-DP mechanism, $\mathcal{M}: \mathcal{D} \rightarrow \mathcal{V}$, is such that, for every $d \in \interior{B}$, 
\begin{align*}
    R_{d} &= R_{\text{Ini}} (n_B, \dist(d,\partial B)) \quad \text{if} \quad \dist(d,\partial B) \leq \tau +1 , \\
    R_{d} &= R_{\text{Ter}} (n_B-\tau-1, \dist(d,\partial B)-\tau-1) \quad \text{otherwise}.
\end{align*}

\end{theorem}

\begin{proof}
Follows directly from combining Theorems \ref{theo: general to line} and \ref{theo: optimal line mech}.
\end{proof}

Thus, if we consider the colored graph on the left hand side of Fig.~\ref{fig: boundary graph} subject to $R_h = R_\ell = 1/13$, then, the optimal mechanism is such that $R_a = R_b= 0$, $R_c=R_g=R_k = 0$, $R_i = R_n = R_q = 0.2$, $R_d = R_o =R_m = 0.36$, $R_j =R_p =R_r=0.568$, and $R_e = R_f = 0.7447$. We note that this mechanism can be obtained by pulling back the optimal mechanism for the $(3,4)$-line in Example~\ref{ex: (4,3)-line}.

We now show that when the probability in the boundary of the blue set is such that the output blue is more likely, the optimal mechanism depends only on the terminal recurrence.

\begin{corollary} \label{cor: half}
Consider the setting in Theorem \ref{theo: general graph}. If the boundary probability $R_{n_B} \leq \frac{1}{2}$, then $R_{d} = R_{\text{Ter}} (n_B , \dist(d,\partial B))$, i.e., for every $d \in \interior{B}$,
\begin{align*}
    R_d = \frac{R_{n_B}}{e^{\epsilon \dist(d,\partial B)}} - \frac{\delta (e^{\epsilon \dist(d,\partial B)}-1)}{e^{\epsilon \dist(d,\partial B)}(e^\epsilon-1)}.
\end{align*}

\end{corollary}

\begin{proof}
This follows from Lemma \ref{lem: ineq} by substituting $x=R_n$, $y=e^\epsilon$, $z=\delta$, and noting that the inequality in the lemma implies that the optimal bound at each step is given in \eqref{eq: dp2}. Alternatively, one can show that, in this case, $\tau < 0$.
\end{proof}

A particular case of boundary homogeneity is when the mechanism gives no preference for blue or red at the boundary.

\begin{definition}
A mechanism $\mathcal{M}: \mathcal{D} \rightarrow \mathcal{V}$ is balanced (or fair) if $B_d = R_{d'}$ for every $d \in \partial B$ and $d' \in \partial R$.
\end{definition}

In the case of balanced mechanisms, the optimal mechanism takes the following simple form.

\begin{corollary}
The optimal balanced $(\epsilon,\delta)$-DP mechanism is such that, for every $d \in \interior{B}$,
\begin{align*}
    R_d = \frac{e^\epsilon - 1 - \delta (e^{\epsilon (\dist(d,\partial B)+1)} +e^{\epsilon \dist(d,\partial B)} -2)}{e^{\epsilon \dist(d,\partial B)} (e^\epsilon +1)(e^\epsilon -1)}.
\end{align*}
\end{corollary}

\begin{proof}
Let $x = B_d = R_{d'}$ for every $d \in \partial B$ and $d' \in \partial R$. Then, the $(\epsilon,\delta)$-DP conditions are equivalent to $x \leq e^\epsilon (1-x)+\delta$ and $(1-x) \leq e^\epsilon x + \delta$, of which only the first equation gives an upper bound on $x$. Maximizing $x$, we obtain $B_d = x = \frac{e^\epsilon + \delta}{1 + e^\epsilon}$ which implies in the boundary $R_d = \frac{1-\delta}{1+e^\epsilon}$, for every $d \in \partial B$. Since $R_d \leq \frac{1}{2}$, the result follows from Corollary~\ref{cor: half}.
\end{proof}

Thus, if we consider the voting example in Fig.~\ref{fig: ex graphs a}, the optimal balanced $(\log(2),0.1)$-DP mechanism is such that $R_{111}=0.1$, $R_{211} = R_{121} = R_{112}= 0.3$, $R_{122} = R_{212} = R_{221}= 0.7$, and $R_{222} = 0.9$.

\appendix

In this Appendix, we prove Lemmas 1 through 4 used in the results of the main text. To apply them to the main results we generally substitute the variable $x$ by the probability $R_n$, the variable $y$ by the privacy parameter $\epsilon$, and the  variable $z$ by the privacy parameter $\delta$.

The first lemma we present shows the conditions under which \eqref{eq: dp1>dp2 ineq in proof} is true in Theorem \ref{theo: optimal line mech}.

\begin{lemma} \label{lem: dp1>dp2}
Let $x,y,z \in \mathbb{R}$ be such that $0\leq x \leq 1$, $y \geq 1$, and $0\leq z < 1$. Then, 
\begin{align} \label{eq: lem dp1>dp2}
    1 - y + xy - z > \frac{x-z}{y}
\end{align}
if and only if $x=1$, $y>1$, or $z <xy +x-y$.
\end{lemma}

\begin{proof}
If $x=1$, then \eqref{eq: lem dp1>dp2} takes the form $1-z>\frac{1-z}{y}$. Since $y \geq 1$, this is equivalent to $y(1-z)>1-z$. But this occurs if and only if $y>1$ and $0\leq z < 1$. 

If $x \neq 1$, then \eqref{eq: lem dp1>dp2} is equivalent to $y-y^2+xy^2-zy>x-z$. Putting all terms on one side and dividing by $(x-1)<0$ we obtain
\begin{align*}
    \frac{y-y^2+xy^2-zy-x+z}{x-1}<0 ,
\end{align*}
which can be factored into
\begin{align*}
    \left( y - 1 \right) \left( y - \frac{x-z}{1-x} \right) < 0 .
\end{align*}
Since $y>1$, this is equivalent to $y < \frac{x-z}{1-x}$, which is equivalent to $z <xy +x-y$.
\end{proof}

The next lemma solves the recurrence in Theorem \ref{theo: optimal line mech} used to define the initial recurrence, $R_{\text{Ini}}$, in Definition \ref{def: initial recurrence}.

\begin{lemma} \label{lem: recur DP1}
Let $n \in \mathbb{N}$ and let $x_1, \ldots, x_n,y,z \in \mathbb{R}$ such that $y>1$ satisfy  
\begin{align} \label{eq: lem: recur DP1 recursion}
    x_i = 1+yx_{i+1}-y-z \quad \text{for every $i < n$}.
\end{align}
Then, for every $i < n$,
\begin{align}
    x_{i} &= 1+y^{n-i}(x_n-1) - z (1+y+\ldots+y^{n-i-1}) \nonumber \\
    &= 1+y^{n-i}(x_n-1) - \frac{z(1-y^{n-i})}{1-y} . \label{eq: lem: recur DP1 solution}
\end{align}
\end{lemma}

\begin{proof}
Since \eqref{eq: lem: recur DP1 recursion} is a linear recursion it has a unique solution, which can be verified by substituting \eqref{eq: lem: recur DP1 solution} in \eqref{eq: lem: recur DP1 recursion}.
\end{proof}

The following lemma solves the recurrence in Theorem \ref{theo: optimal line mech} used to define the terminal recurrence, $R_{\text{Ter}}$, in Definition \ref{def: terminal recurrence}.

\begin{lemma} \label{lem: recur DP2}
Let $n \in \mathbb{N}$ and let $x_1, \ldots, x_n,y,z \in \mathbb{R}$ such that $y>1$ satisfy 
\begin{align} \label{eq: lem: recur DP2 recursion}
    x_i = \frac{x_{i+1} - z}{y} \quad \text{for every $i < n$}.
\end{align}
Then, for every $i < n$,
\begin{align}
    x_{i} &= \frac{x_n - z(1+y+\ldots+y^{n-i-1})}{y^{n-i}} \nonumber \\
    &= \frac{x_n}{y^{n-i}} - \frac{z (1-y^{n-i})}{(1-y)y^{n-i}} \label{eq: lem: recur DP2 solution}
\end{align}
\end{lemma}

\begin{proof}
Since \eqref{eq: lem: recur DP2 recursion} is a linear recursion, it has a unique solution which can be verified by substituting \eqref{eq: lem: recur DP2 solution} in \eqref{eq: lem: recur DP2 recursion}.
\end{proof}

Our final lemma is used in the proof of Corollary \ref{cor: half}.

\begin{lemma} \label{lem: ineq}
Let $x,y,z \in \mathbb{R}$ be such that $x \leq \frac{1}{2}$, $y \geq 1$, and $z \geq 0$. Then, it holds that
\begin{align} \label{eq: lem ineq}
    1-y+xy-z \leq \frac{x-z}{y} .
\end{align}
\end{lemma}

\begin{proof}

Since $x\leq 1/2$ and $y \geq 1$, it follows that $xy + x \leq \frac{y+1}{2}$. But $y \geq 1$ implies in $\frac{y+1}{2} \leq y$. Thus, $xy + x \leq y$. We rewrite this as $0 \leq \left( y - \frac{x}{1-x} \right)$. Since $y \geq 1$ it follows that $0 \leq \left( y - \frac{x}{1-x} \right) \left( y-1 \right)$. Expanding this equation we obtain $y+xy^2 \leq x+y^2$. Now, since $y \geq 1$ and $z \geq 0$, it follows that $z \leq zy$. Thus, $y+xy^2 + z \leq x+y^2 +zy$. Rearranging this equation, we obtain \eqref{eq: lem ineq}.

\end{proof}

\bibliographystyle{IEEEtran}
\bibliography{ref.bib}

\end{document}